\documentclass[11pt]{article}
\usepackage{fullpage}
\usepackage{epsfig}
\usepackage{graphics}
\usepackage{latexsym}
\usepackage{amsmath}
\usepackage{amsfonts}
\usepackage{amssymb}
\usepackage{mathrsfs}
\usepackage{amsthm}
\usepackage{xspace}
\usepackage{epstopdf}
\usepackage{float}
\usepackage{hyperref}

\numberwithin{equation}{section}
\usepackage[ruled,vlined]{algorithm2e}
\SetArgSty{textrm}

\usepackage[english]{babel}
\usepackage[nottoc]{tocbibind}

\usepackage{caption}
\usepackage{subcaption}

\usepackage[usenames,dvipsnames]{color}
\usepackage{pifont}
\usepackage{booktabs}
\usepackage{multirow}

\usepackage{amsthm}     
\theoremstyle{plain}     
\newtheorem{theorem}{Theorem}

\newtheorem{lemma}{Lemma}
\newtheorem{proposition}{Proposition}
\theoremstyle{definition} 

\newtheorem{mechanism}{Mechanism}
\theoremstyle{remark} 

\makeatletter
\@addtoreset{equation}{section}
\def\section{\@startsection {section}{1}{\z@}{-3.5ex plus -1ex minus
 -.2ex}{2.3ex plus .2ex}{\large\bf}}
\makeatother

\def\bfm#1{\mbox{\boldmath$#1$}}

\def\0{\bfm 0}

\DeclareMathAlphabet{\mathpzc}{OT1}{pzc}{m}{it}

\newcounter{my}

\newcounter{my2}

\newcounter{my3}

\newcounter{my4}

\newcounter{my5}

\newcounter{my6}

\allowdisplaybreaks 

\begin{document}

\title{Mechanism Design for Extending the Accessibility of Facilities}


\date{}
\maketitle

\vspace{-3em}
\begin{center}

\author{Hau Chan$^{1}$\quad Jianan Lin$^{2}$\quad Chenhao Wang $^{3,4}$ \quad Yanxi Xie$^{5}$\\
${}$\\
$1$ University of Nebraska-Lincoln\\
$2$ Rensselaer Polytechnic Institute\\
$3$ Beijing Normal University-Zhuhai chenhwang@bnu.edu.cn\\
$4$ BNU-HKBU United International College\\
$5$ School of Artificial Intelligence, Beijing Normal University\\
\medskip
}

\end{center}

\begin{abstract}
We study a variation of facility location problems (FLPs) that aims to improve the accessibility of agents to the facility within the context of mechanism design without money. 
In such a variation, agents have preferences on the ideal locations of the facility on a real line, and the facility’s location is fixed in advance where (re)locating the facility is not possible due to various constraints (\(e.g.,\) limited space and construction costs). 
To improve the accessibility of agents to facilities, existing mechanism design literature in FLPs has proposed to structurally modify the real line (e.g., by adding a new interval) or provide shuttle services between two points when structural modifications are not possible. 
In this paper, we focus on the latter approach and propose to construct an accessibility range to extend the accessibility of the facility. 
In the range, agents can receive accommodations (e.g., school buses, campus shuttles, or pickup services) to help reach the facility. 
Therefore, the cost of each agent is the distance from their ideal location to the facility (possibility) through the range.
We focus on designing strategyproof mechanisms that elicit true ideal locations from the agents and construct accessibility ranges (intervals) to approximately minimize the social cost or the maximum cost of agents. 
For both social and maximum costs, we design group strategyproof mechanisms with asymptotically tight bounds on the approximation ratios. 
\end{abstract}

\section{Introduction}

In recent years, facility location problems (FLPs) and their variants \cite{chan2021mechanismsurvey} have been extensively studied within the context of approximate mechanism design without money \cite{aziz2020facility, dokow2012mechanism, procaccia2013approximate, todo2011false}. 
In the most typical setting of FLPs, a social planner is tasked with locating a facility (\(e.g.,\) a park, hospital, or school) in a metric place (\(e.g.,\) a real line) to serve a set of agents, who have preferences on the ideal locations of the facility, and minimize the distance of the facility to agent ideal locations. 

Because of the strategic behavior of the agents, the agents may have the incentive to misreport their ideal locations to manipulate the facility location to be closer to their ideal locations.
Therefore, the main focus of the study of mechanism design for FLPs is in designing strategyproof mechanisms to incentivize the agents to report their locations truthfully while simultaneously determining the locations of a facility that (approximately) minimizes a given cost objective (\(e.g.,\) minimizing the total/social or maximum distance of the agents to the facility).

\paragraph{Improving the Accessibility to Facilities.}
In many real-world scenarios, facilities (\(e.g.,\) parks, hospitals, or schools) have already been built for many years or decades. 
Because of the changes in local demographics and social planner objectives, the locations of existing facilities may no longer be ideal and could lead to accessibility concerns for the agents. 
One straightforward approach to improve the accessibility of the facility is to "ignore" the existing facility and (re)locate a new facility tailored to the agents. 
However,  such an approach can be impractical because the facility cannot be easily (re)located due to various constraints (e.g., construction costs, limited space, or legal regulations).

To improve the accessibility of agents to facilities, recent mechanism design studies in FLPs \cite{chan2023aamas,chan2024aamas,Zining2024TAMC} have considered variations of FLPs on real lines and proposed two complementary approaches. 
These approaches are to either structurally modify the real lines \cite{chan2024aamas,Zining2024TAMC} or provide free accommodation services between two points when structural modifications are not possible \cite{chan2023aamas,chan2024aamas}. 
In the former approach \cite{chan2024aamas,Zining2024TAMC}, to structurally modify the real lines, the social planner can construct new edges (such as roadways or bridges) between two points so that agents can use the new edges to access the facility more efficiently. 
In the latter approach \cite{chan2023aamas,chan2024aamas}, the accommodation services can refer to some form of free shuttle services between two fixed pickup/drop-off points (via a cost-free short-cut edge) within the real line to decrease the distance agents need to reach the facility. 
Through both approaches to improving accessibility, the agents can reduce their distances to reach the facility more efficiently. 

\paragraph{Our Approach: Extending the Accessibility to Facilities.}
In this paper, we focus on the latter case when structural modifications are not possible due to various reasons (e.g., the time, costs, and legal regulations of modifying the structure) and propose to construct an accessibility range (e.g., zone or radius) to extend the accessibility of the facility.  
Extending from the accommodation services considered by \cite{chan2023aamas,chan2024aamas} where agents can only use the services at two pickup/drop-off points, in the proposed accessibility range, the agents can receive accommodation services (e.g., school buses, campus shuttles, or pickup services) within any points of the ranges to help reach the facility without incurring any additional cost.  For example, consider providing a shuttle service along a range. 
The shuttle can make stops at every point along the range for the agents rather than only two pickup or drop-off points in \cite{chan2023aamas,chan2024aamas} to transport the agents to the facility. Therefore, the cost of the agents between any points within the accessibility range is zero.

The accessibility range of a facility is common in many real-world situations where a social planner needs to determine a range in which they can deploy transportation mediums (like buses or shuttles) to pick up agents along that range and bring the agents into a common drop-off point near or at the facility\footnote{If the range contains the facility, the agents can reach the facility location directly. If the range does not contain the facility, the agents can reach an endpoint of the range and then reach the facility from the endpoint.}. 
For instance, when planning the school bus routes for students, the school district first needs to determine the range for picking up students from their homes to schools to ensure comprehensive accessibility for all students. 
Similarly, when designing (university or workplace) campus shuttle routes, the planners first need to determine the range that covers the campus. 
When providing medical services to patients, different providers could  determine the coverage areas to transport patients to medical facilities. 

\paragraph{Our Contributions.} 
In this paper, we consider a variation of FLPs to improve the accessibility of agents to the facility by constructing an accessibility range within the context of approximate mechanism design without money. 
In particular, we consider the most typical mechanism design setting of FLPs and existing studies \cite{chan2023aamas,chan2021mechanismsurvey,chan2024aamas,Zining2024TAMC} where a facility is located at a predetermined location on a real line $\mathbb R$. A set of $n$ agents is located on $\mathbb R$, and the facility is fixed at $0$ without loss of generality. The accessibility range is represented by an interval $(a,b)$, and the length of the range cannot exceed a constraint constant $d$, i.e., $|a-b|\le d$.

Given the accessibility range $(a,b)$, if the range contains the facility, an agent's \emph{cost} is either the distance of their ideal location to the range or zero (if the agent's ideal location is within the range) because accommodation services between any two points within the range incur zero cost to the agent.
Otherwise, the agent's cost is either the distance of the ideal location to the range (which can be zero if the agent is in the range) and the range's endpoint to the facility or the distance of the ideal location to the facility. 

For example, consider three agents located at $-2,0.8$ and $3$. Given a range $(-1,1)$ that contains the facility at $0$, the cost of the three agents is $1, 0$, and $2$, respectively. Given a range $(1,2)$ that does not contain the facility, the cost of the three agents is $2,0.8$, and $2$, respectively. See Figure \ref{fig:1} for an illustration, where difference agents are depicted in difference colors, and costs are indicated by arrows. 
\begin{figure}[htbp]
    \centering
    \includegraphics[scale=0.6]{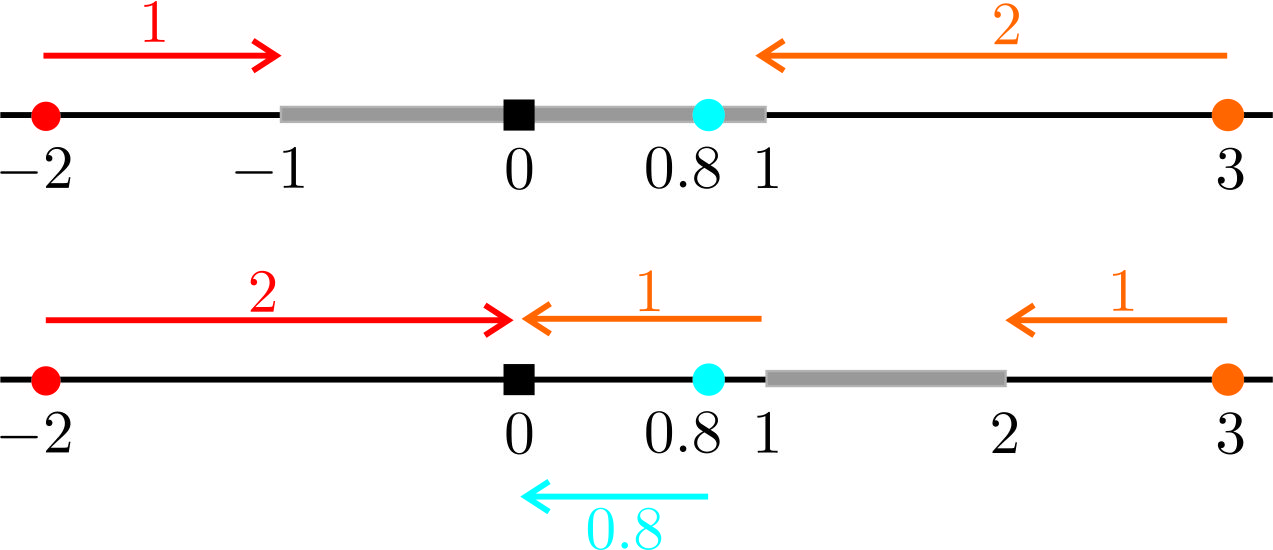}
    \caption{Illustration of agent costs with ranges $(-1,1)$ and $(1,2)$. }\label{fig:1}
\end{figure}

Our goal is to design strategyproof mechanisms that elicit ideal facility locations of agents and determine accessibility ranges (or intervals) within the real line to (approximately) minimize the social cost or the maximum cost of agents. 
Below, we elaborate on our mechanism design results and summarize them in Table \ref{tab:summary}.

\begin{itemize}
\item For the social cost, we show that there is an optimal and (group) strategyproof mechanism. 
 Furthermore, we consider designing strong group strategyproof (SGSP) mechanisms, a stronger notion of group strategyproofness that only requires that some agents within any group have incentives to misreport ideal locations. 
We show that there is an $(n-1)$-approximation SGSP mechanism, and
we complement these results by showing there is a lower bound of $\frac{n+2}{4}$ for the social cost. 

\item For the maximum cost, we show that there is (group) strategyproof mechanism with an approximation ratio of 2 and provide a matching lower bound of 2. 
 Furthermore, we consider strong group strategyproof mechanisms.
We show that there is a 2-approximation SGSP mechanism. 
We complement these results by showing there is a lower bound of $2$. 

\item We further consider randomized strategyproof mechanisms. For the social and maximum costs,  no randomized SGSP mechanism has better approximation ratios than $\frac{n+2}{4}$ and 1.5, respectively.
\end{itemize}

\begin{table}[h]
\caption{A Summary of Our Results } \label{tab:summary}
\centering
\begin{tabular}{ccc}
\toprule
Objective & SP\;/\;GSP & SGSP \\
\midrule
\multirow{2}{*}{Social cost} & UB: 1 & UB: $n-1$ \\
                            & LB: 1 & LB: $\dfrac{n+2}{4}$ \\
\midrule
\multirow{2}{*}{Maximum cost} & UB: 2 & UB: 2 \\
             & LB: 2 &  LB: 2 \\
\bottomrule
\\
\end{tabular}

SP = strategyproof, GSP = group SP, SGSP = strong GSP

UB = upper bound, LB = lower bound
\end{table}

{While we do not have any randomized strategyproof mechanisms, the randomized lower bound results show that our deterministic GSP/SGSP mechanisms are reasonable even when considering randomization.}

\paragraph{Outline.} { In Section \ref{sec:model}, we present the model of the considered variation of FLPs. In Section \ref{sec:social} and Section \ref{sec:max}, we consider the social cost and the maximum cost, respectively. In Section \ref{sec:ext}, we discuss randomized strategyproof mechanisms.   }

\paragraph{Related Work.} 
In the mechanism design studies of FLPs, Procaccia and Tennenholtz \cite{procaccia2013approximate} are the first to use the FLPs as case studies to demonstrate the concept of approximate mechanism design without money, investigating the design of approximately optimal strategyproof mechanisms through the lens of approximation ratios. 
They consider the settings of locating one or two facilities on the real line and the social and maximum costs. 
For locating one facility, they provide tight bounds on the approximation ratios for strategyproof mechanisms for both objectives. 
Building on this direction, subsequent mechanism design studies have thoroughly investigated a wide range of variations of FLPs, including preferences for facilities \cite{fong2018facility, fotakis2013winner, PaoloCarmine2016}, distance constraints \cite{chen2018mechanism, chen2021tight}, various cost functions \cite{alon2010strategyproof, feigenbaum2017approximately, feldman2013strategyproof, fotakis2013strategyproof}, and mechanism characterizations\cite{lin2020nearly, tang2020characterization}. We refer readers to a recent survey in \cite{chan2021mechanismsurvey}. 
However, different from our work, these mechanism design studies consider locating facilities. 

Our work focuses on the situations in which facilities have been located where relocating them is no longer possible. 
Previous mechanism design studies have proposed to improve the accessibility to a facility by either modifying the real lines \cite{chan2024aamas,Zining2024TAMC} or providing free accommodation services between two points when structural modifications are not possible \cite{chan2023aamas,chan2024aamas}. 
The work of \cite{chan2023aamas} is the first to consider providing free accommodation services between two points from the perspective of mechanism design without money.  
More specifically, \cite{chan2023aamas} aim to design strategyproof mechanisms to elicit ideal locations from agents and determine a zero-cost shortcut between two points to approximately minimize the social/total or maximum cost.  
Considering structural modification within the paradigm of mechanism design without money, \cite{chan2024aamas} further consider adding a short-cut edge with a cost proportional to the distance of its two endpoints divided by a discount factor, as well as adding two zero-cost shortcut edges (extending the setting of \cite{chan2023aamas}). The work of \cite{Zining2024TAMC} proposes to consider designing mechanisms to construct bridges to connect two regions separated by a physical barrier.  
However, all these above-mentioned studies focus either on the perspective of structural modification or constructing zero-cost short-cut edges (i.e., accommodation services) that can be accessed at only two points. In our setting, the agents can receive (zero-cost) accommodation services at any point within the accommodation range of the facility. 
{Different from our work, in \cite{chan2023aamas,chan2024aamas}, the cost of an agent is either the distance of the ideal location to the facility directly or the distances of the ideal location to an endpoint of the range and the range's endpoint to the facility regardless of whether the facility or the agent is within the range.}

\section{Preliminaries}\label{sec:model}

There are $n$ agents $N = \{1, 2, ..., n\}$ on a real line. 
Each agent $i\in N$ has an ideal location $x_i\in\mathbb R$. Denote the agents' location profile by $\mathbf{x} = (x_1, x_2, \ldots, x_n)\in \mathbb{R}^n$. 
We want to determine the accessibility range $(a,b)\in \mathbb{R}^2$ of the facility with fixed location at $0$. This range can be regarded as an interval of $\mathbb R$ and has to satisfy a distance constraint $|a-b|\le d$, where $d$ is a constant. We always assume w.l.o.g. that $a\le b$. 

Given $(a,b)$, the agents can travel within the interval with zero cost. 
{As mentioned earlier, the reason the cost is zero is that the accommodation services are provided to the agents without any charge. 
The agents also do not need to exert any direct efforts to travel within the range explicitly (e.g., shuttle services).}
{We can also view the range to be} equivalent to the case when we shrink the interval into a single point. 
Moreover, each agent $i\in N$ has a \emph{cost} $c(x_i,a,b)$ equal to the distance from their ideal location $x_i$ to the facility after shrinking the interval $(a,b)$.

A (deterministic) mechanism $f: \mathbb{R}^n \rightarrow \mathbb{R}^2$ maps a location profile $\mathbf{x}$ to the accessibility range $(a, b)$. 
A mechanism is \emph{strategyproof} (SP) if an agent can never
benefit from reporting a false location, regardless of the
strategies of others.
Formally, a mechanism $f$ is SP if for all $\mathbf x\in\mathbb R^n,i\in N, x_{i}'\in\mathbb R$, we have $c(x_i,f(\mathbf x))\le c(x_i,f(x_i',\mathbf x_{-i}))$, {where $\mathbf x_{-i}$ is the location profile of all agents but $i$}.
Further, a mechanism is \emph{group strategyproof} (GSP) if no group of agents can collude to misreport in a way that makes every member better off. Formally,  $f$ is GSP if for all $\mathbf x\in\mathbb R^n,S\subseteq N, \mathbf x_{S}'\in\mathbb R^{|S|}$, there exists $i\in S$ so that $c(x_i,f(\mathbf x))\le c(x_i,f(\mathbf x_S',\mathbf x_{-S}))$. Clearly, a GSP mechanism must be SP. 

Another well-studied notion in mechanism design and social choice theory is the \emph{strong group strategyproofness} (SGSP) \cite{manjunath2014efficient}, indicating that no group of agents can collude to misreport so that no member is worse off and at least one is better off. Formally,  $f$ is SGSP if for all $\mathbf{x} \in \mathbb{R}^n$, there does not exist a group of agents $S\subseteq N$ who misreport their location to $\mathbf{x}_S'$, so that   $c(x_i, f(\mathbf x_S',\mathbf x_{-S})) \le c(x_i, f(\mathbf{x}))$ for $\forall i\in S$, and  $c(x_j, f(\mathbf x_S',\mathbf x_{-S})) < c(x_j, f(\mathbf{x}))$ for some $j\in S$.
Clearly, an SGSP mechanism must be GSP. 

We consider two objectives, minimizing the \emph{social cost} $\text{SC}(\mathbf x, a,b)$, and minimizing the \emph{maximum cost} $\text{MC}(\mathbf x, a,b)$:
$$
\left\{
\begin{aligned}
    \text{SC}(\mathbf x, a,b)&=\sum_{i\in N}c(x_i,a,b)\\
    \text{MC}(\mathbf x, a,b)&=\max_{i\in N}c(x_i,a,b)
\end{aligned}
\right.
$$
A mechanism $f$ is $r$-approximation for objective $T\in\{\text{SC},\text{MC}\}$ if 
$$\sup_{\mathbf x}\frac{T(\mathbf x,f(\mathbf x))}{\min\limits_{a,b:|a-b|\le d}T(\mathbf x,a,b)}\le r.$$ 

\medskip
We give two {intuitive lemmas} on the accessibility ranges. The first one says that if the range does not cover the facility at 0, then we can move it so that its endpoint reaches 0. The second one says that the length of the range should be as large as possible. 

\begin{lemma}\label{lemma:basic1}
For any range $(a, b)$ with $a > 0$ and agent $i\in N$,  $c(x_i,a,b)\ge c(x_i,0,b-a)$.
For any range $(a, b)$ with $b< 0$ and agent $i\in N$,  $c(x_i,a,b)\ge c(x_i,a-b,0)$.
\end{lemma}
\begin{proof} 
We only prove the case when $a>0$.    If $x_i \le 0$, the cost of $i$ under intervals $(a, b)$ and $(0, b - a)$ are both equal to $|x_i|$. If $0 < x_i \le b - a$, the cost of $i$ under $(0, b-a)$ is 0, indicating that it cannot be larger than the cost under $(a, b)$. If $x_i > b - a$, the cost of $i$ under $(0, b-a)$ is $x_i - (b - a)$,  while the cost  under $(a, b)$ is $x_i$ when $x_i \le a$, and $x_i - \left(\min\left(x_i, b\right)-a\right) \ge x_i - (b - a)$ when $x_i > a$. 
\end{proof}

\begin{lemma}\label{lemma:basic2}
 For any range $(a,b)$ with $b - a < d$ and any agent $i\in N$, $c(x_i,a,b)\ge c(x_i,a,a+d)$. 
\end{lemma}
\begin{proof}
It is clear that by increasing the length of the interval with one endpoint fixed,  no agent would increase their cost. 
\end{proof}

Given location profile $\mathbf x$,  let $N_1=\{i\in N|x_i\le 0\}$ be the set of agents on the left of the facility, and $N_2=N\setminus N_1$ the set of agents on the right. Let $x_l=\min\{x_i|i\in N\}$ be the leftmost agent location, and $x_r=\max\{x_i|i\in N\}$ be the rightmost agent location.

\section{The Social Cost}\label{sec:social}
We study the social cost.  
In  Section \ref{sec:socgsp}, we present an optimal GSP mechanism. 
In Section \ref{sec:socsgsp}, we provide (asymptotically) tight upper and lower bounds on the approximation ratios of SGSP mechanisms. 

\subsection{GSP Mechanisms}\label{sec:socgsp}

When all agents are located on the left {of the facility}, it is clear that $(-d,0)$ is an optimal solution since every agent achieves their best possible cost. Similarly, when all agents are located on the right {of the facility}, the solution $(0,d)$ is optimal. 

The case when $x_l<0<x_r$ is more complicated. For any point $y\in\mathbb R$, define $N_l(y)=\{i\in N|x_i\le y\}$ to be the set of agents who are to the left of $y$, and $N_r(y)=\{i\in N|x_i\ge y\}$ to be the set of agents who are to the right of $y$. Let $D=\{y\in\mathbb R| N_l(y)\ge N_r(y+d) \}$ be the set of points $y$ on the real line so that the number of agents to the left of $y$ is at least the number of agents to the right of $y+d$. We use $\inf D$ to denote the infimum of $D$.

\begin{mechanism}\label{mec:sc1}
    Given location profile $\mathbf x$:
 
    \quad if $\inf D\ge 0$, return $(0,d)$;
 
    \quad if $\inf D\le -d$, return $(-d, 0)$;
 
    \quad if $-d < \inf D < 0$, return $(\inf D, \inf D + d)$. 
\end{mechanism}

\begin{lemma}\label{lemma:sc1-1}
    Mechanism \ref{mec:sc1} returns an optimal solution for the social cost.
\end{lemma}

\begin{proof}
    We discuss the optimality in the three cases.

    \textbf{Case 1}. $\inf D\ge 0$. By Lemmas \ref{lemma:basic1} and \ref{lemma:basic2}, it is clear that $(0,d)$ is the best one among all solutions $(a,b)$ with $a\ge 0$. We only need to compare with any solution $(a,a+d)$ with $-d\le a<0$, because the social cost of  $(a, a+d)$ with $a<-d$ cannot be less than that of $(-d, 0)$. We have
    \begin{align*}
        & \text{SC}(\mathbf x, a,a+d) - \text{SC}(\mathbf x, 0, d) \\
        =~& \sum_{i:x_i < 0}\max(0, -x_i + a) + \sum_{i:x_i \ge 0}\max(0, x_i - (a + d)) - \sum_{i:x_i < 0}|x_i| - \sum_{i: x_i \ge 0}\max(0, x_i - d) \\
        \ge~& \sum_{i: x_i < 0}a + \sum_{i: 0 \le x_i < d} (0-0) + \sum_{i:x_i \ge d}(-a) = |i\in N_l(0): x_i < 0|\cdot a - |N_r(d)|\cdot a.
    \end{align*}
    Since $\inf D\ge 0$, we know $|i\in N_l(0): x_i < 0| \le |N_r(d)|$, as otherwise we can find a small positive $\epsilon \rightarrow 0$ so that $N_l(-\epsilon) \ge N_r(d-\epsilon)$, as long as there is no agent at $[\epsilon, 0)$ or $[d-\epsilon, d)$. 
    Therefore we have 
    $$\text{SC}(\mathbf x,a,a+d) - \text{SC}(\mathbf x,0, d) \ge 0,$$
    indicating that $(0,d)$ is optimal for the social cost. 

    \textbf{Case 2}.  $\inf D\le -d$. By Lemmas \ref{lemma:basic1} and \ref{lemma:basic2}, it is clear that $(-d,0)$ is the best one among all solutions $(a,b)$ with $b\le 0$. We only need to compare with any solution $(b-d, b)$ with $0 < b \le d$, because the social cost of $(b-d, b)$ with $b>d$ cannot be less than that of $(0, d)$. We have
    \begin{align*}
        & \text{SC}(\mathbf x, b-d,b) - \text{SC}(\mathbf x, -d, 0) \\
        =~& \sum_{i: x_i > 0}\max(0, x_i - b) + \sum_{i:x_i \le 0}\max(0, -x_i + b - d)) - \sum_{i: x_i > 0}x_i - \sum_{i: x_i \le 0}\max(0, -x_i - d) \\
        \ge~& \sum_{i: x_i > 0}-b + \sum_{i: -d < x_i \le 0} (0-0) + \sum_{i:x_i \le -d}b = N_l(-d)\cdot b  -  |i\in N_r(0): x_i > 0|\cdot b.
    \end{align*}
    Since $\inf D\le -d$, we know $|i\in N_r(0): x_i > 0| \le |N_l(-d)|$, indicating that $(-d,0)$ is optimal for the social cost.

    \textbf{Case 3}. $-d < \inf D < 0$.  The mechanism returns $(\inf D, \inf D + d)$. By Lemmas \ref{lemma:basic1} and \ref{lemma:basic2}, we only need to compare with $(a, a + d)$ with $-d \le a < \inf D$ and $\inf D < a \le 0$, respectively. If $a<\inf D$, we have
    \begin{align*}
        & \text{SC}(\mathbf x, a,a+d) - \text{SC}(\mathbf x, \inf D, \inf D + d) \\
        =~& \sum_{i: x_i \ge 0}\max(0, x_i - a - d) + \sum_{i: x_i < 0}\max(0, -x_i + a) - \sum_{i: x_i \ge 0}\max(0, x_i - \inf D - d) \\
        &- \sum_{i: x_i < 0}\max(0, -x_i + \inf D) \\
        \ge~& \sum_{i: x_i \ge \inf D+d}(\inf D - a) + \sum_{i: \inf D \le x_i < \inf D + d} (0-0) + \sum_{i: x_i < \inf D}(a - \inf D)\\
        =~& (\inf D - a)\cdot (|N_r(\inf D + d)|-|\{i\in N_l(0) : x_i < \inf D\}|)
        \ge 0.
    \end{align*}

    In the same way, if $a>\inf D$, we have
    \begin{align*}
        & \text{SC}(\mathbf x, a,a+d) - \text{SC}(\mathbf x, \inf D, \inf D + d) \\
        \ge ~& (a - \inf D)\cdot (|N_l(\inf D)|-|\{i\in N_r(0): x_i > \inf D + d\}|)\\
        \ge ~& 0.
    \end{align*}
In both cases, it follows that
    $$
    \text{SC}(\mathbf x,a,a+d) - \text{SC}(\mathbf x, \inf D, \inf D + d) \ge 0.
    $$
Therefore, the solution $(\inf D, \inf D + d)$ returned by the mechanism is optimal. 
\end{proof}

Next, we show the group strategyproofness of this mechanism. 

\begin{lemma}\label{lemma:sc1-2}
    Mechanism \ref{mec:sc1} is group strategyproof.
\end{lemma}

\begin{proof}
    We discuss the misreporting of agent group $A\subseteq N$ in the three cases.

   \textbf{Case 1}. $\inf D\ge 0$. The mechanism returns $(0, d)$. It is easy to see that any agent $i$ with $x_i\ge 0$ has achieved their minimum possible cost and thus has no incentive to misreport. Therefore, the group $A$ consists of some agents to the left of 0. However, no matter how this group of agents misreports, it would always be the case $\inf D\ge 0$, and the outcome would not change. Thus, this group cannot benefit from misreporting. 

   \textbf{Case 2}. $\inf D\le -d$.  The mechanism returns $(-d, 0)$. It is easy to see that any agent $i$ with $x_i\le 0$ has achieved their minimum possible cost and thus has no incentive to misreport. Therefore, the group $A$ consists of some agents to the right of 0. However, no matter how this group of agents misreports, it would always be the case $\inf D\le -d$, and the outcome would not change. Thus, this group cannot benefit from misreporting. 

   \textbf{Case 3}. $-d < \inf D < 0$.  The mechanism returns $(a, b) = (\inf D, \inf D + d)$. It is easy to see that any agent $i$ with $\inf D\le x_i\le \inf D + d$ has achieved their minimum possible cost and thus has no incentive to misreport.
   Therefore, all agents in the group should have either $x_i > \inf D + d$ or $x_i < \inf D$. We note that if agent $i$ with $x_i < \inf D$ decreases the cost, then it must be the case that $\inf D$ moves to the left after misreporting. If agent $i$ with $x_i > \inf D + d$ decreases the cost, then it must be the case that $\inf D$ moves to the right after misreporting. It indicates that the group is either a subset of $\{i\in N| x_i < \inf D\}$ or a subset of $\{i\in N| x_i > \inf D + d\}$. However, the agents in $\{i\in N| x_i < \inf D\}$ cannot move $\inf D$ to the left by misreporting, and the agents in $\{i\in N| x_i > \inf D + d\}$ cannot move $\inf D$ to the right. Hence, no group of agents can misreport so that every group member gains. 
\end{proof}

The above two lemmas immediately give the following theorem. 
 
\begin{theorem}\label{theorem:1}
   Mechanism \ref{mec:sc1} is group strategyproof and optimal for the social cost.
\end{theorem}

\subsection{SGSP Mechanisms}\label{sec:socsgsp}

We note that Mechanism \ref{mec:sc1} is not strong group strategyproof. Suppose that $n\ge 4$ is an even number, and $d=1$. 
 The location profile contains one agent at $-d$, $\frac{n}{2}-1$ agents at $-\epsilon$ where $\epsilon$ is a sufficiently small positive number, and $\frac{n}{2}$ agents at $d$. 
The infimum of $D$ is $-\epsilon$, and the mechanism returns $(-\epsilon,d-\epsilon)$. The cost of the agent at $-d$ is $c(-d,-\epsilon,\epsilon+d)=d-\epsilon$, and the cost of the agent at $-\epsilon$ is $c(-\epsilon,-\epsilon,\epsilon+d)=0$.s
If all agents at $-\epsilon$ misreport location to $-d$, then the infimum of $D$ becomes $-d$, and the mechanism returns $(-d,0)$. Thus, the agent at $-d$ will have a smaller cost, and all agents at $-\epsilon$ still have 0 cost. Therefore, 
Mechanism \ref{mec:sc1} is not SGSP.

Next, we introduce an SGSP mechanism. The explanation for Mechanism \ref{mec:sc1} being not SGSP is that the value of $\inf D$ is vulnerable to the manipulation of a group of agents. Instead of using $\inf D$, in the following mechanism, we use $x_l$, the leftmost agent location, to distinguish the different cases. We will show that the leftmost agent location can strongly resist the manipulations from such groups, but this approach is not able to achieve a good guarantee of the social cost anymore. 

\begin{mechanism}\label{mec:soc-strong}
    Given location profile $\mathbf x$,
    
    \quad if $x_l \ge 0$, return $(0, d)$;

    \quad if $x_l \le -d$, return $(-d, 0)$;
    
    \quad if $-d< x_l < 0$, return $(x_l, x_l + d)$.
\end{mechanism}

\begin{lemma}\label{lemma:strong1}
    Mechanism \ref{mec:soc-strong} is strong group strategyproof.
\end{lemma}

\begin{proof}
We discuss the three cases and show that no group of agents can misreport simultaneously so that at least one member gains and no member harms. 
    If $x_l\ge 0$, every agent attains their minimum possible cost and has no incentive to misreport together.

    If $x_l \le -d$, the outcome changes only if the leftmost agent location $x_l$ moves to the right of $-d$, and the agent located $x_l$ must be in the group. However, the new outcome would strictly increase the cost of this agent. 

    If $-d< x_l < 0$, the mechanism returns $(x_l,x_l+d)$. Suppose a group of agents misreport so that the output becomes $(a, a+d)$. If $a < x_l$, then no agent can decrease their cost. If $a> x_l$, then the agent at $x_l$ must be in the group; however, the cost of this agent will increase.
\end{proof}

\begin{theorem}
    Mechanism \ref{mec:soc-strong} is strong group strategyproof and $(n-1)$-approximation for the social cost.
\end{theorem}
\begin{proof}
 The strong group strategyproofness is proven in Lemma \ref{lemma:strong1}. We consider the approximation ratio for the social cost. Obviously when $x_l \ge 0$, or $x_r\le 0$, or $x_r - x_l\le d$, the solution is always optimal. We only consider the other two cases with $x_l < 0$ and $x_r - x_l > d$. Denote by $\text{OPT}$ the optimal social cost. It is easy to see that the optimal social cost is at least $x_r-x_l-d$.

    \textbf{Case 1}. $x_l < -d$. The mechanism returns $(-d, 0)$. We have
    \begin{align*}
        \text{SC}(\mathbf x,-d,0) &\le (n-1)(-d-x_l) + (n-1)x_r\\
        &= (n-1)(x_r - x_l -d)\\
        &\le (n-1)\cdot\text{OPT}.
    \end{align*}

    \textbf{Case 2}. $-d\le x_l \le 0$. The mechanism returns $(x_l, x_l + d)$. The cost of agents located in the interval $[x_l,x_l+d]$ is 0,
    and the cost of any agent is at most $x_r-(x_l+d)$. Then we have
    \begin{align*}
       \text{SC}(\mathbf x,x_l,x_l+d) &\le (n-1)(x_r-x_l-d)\\
        &\le (n-1)\cdot\text{OPT}.
    \end{align*}

     Therefore, in both cases, the social cost derived by Mechanism \ref{mec:soc-strong} is at most $n-1$ times the optimal social cost.
\end{proof}

The analysis of the approximation ratio of Mechanism  \ref{mec:soc-strong}  is tight. Consider the location profile where one agent is located at $-d$ and the other $n-1$ are located at $d$. Since $x_l=-d$, the output of the mechanism is $(-d, 0)$, which induces a social cost equal to $(n-1)d$. The optimal solution is $(0,d)$, and the optimal social cost is $d$.  Thus, the social cost induced by the mechanism is $n-1$ times the optimum.

Although this linear approximation ratio $n-1$ looks large, we show that it is indeed the asymptotically best possible result we could expect to achieve for all SGSP mechanisms. We give the following inapproximability result, which excludes the possibility of designing a mechanism with a sublinear approximation ratio.

\begin{theorem}
    No deterministic strong group strategyproof mechanism has an approximation ratio less than $\frac{n+2}{4}$ for the social cost. 
\end{theorem}
\begin{proof}
Let $f$ be a deterministic strong group strategyproof mechanism with an approximation ratio less than $\frac{n+2}{4}$. We consider three location profiles $\mathbf{x}_1, \mathbf{x}_2, \mathbf{x}_3$ with 
even number $n$ of agents. In $\mathbf{x}_1$, $\frac{n}{2}$ agents are located at $-d$ and the other $\frac{n}{2}$ are located at $d$. In $\mathbf{x}_2$, one agent is located at $-d$, $\frac{n}{2}-1$ agents are located at $-\epsilon$ with sufficiently small $\epsilon>0$, and other $\frac{n}{2}$ are located at $d$. In $\mathbf{x}_3$, one agent is located at $d$, $\frac{n}{2}-1$ agents are located at $\epsilon$ and other $\frac{n}{2}$ are  at $-d$. 

We consider $f(\mathbf{x}_1)=(a,b)$ and discuss three cases.

\textbf{Case 1}. $a\le 0\le b$. Consider profile $\mathbf x_3$ and the solution returned by $f$ is $f(\mathbf x_3)=(a',b')$. By the approximation ratio, it is known that $a'<0$ and $b'\le d$. If $|[0,b]|<|[0,b']|$, then the $\frac n2$ agents located at $d$ in $\mathbf x_1$ can collude to misreport so that the profile becomes $\mathbf x_3$ and the solution becomes $(a',b')$. Then, the agents at $d$ in $\mathbf x_1$ strictly decrease their cost, giving a contradiction. If $|[0,b]|>|[0,b']|$, then the $\frac n2$ agents located at the positive half axis in $\mathbf x_3$ can collude to misreport so that the profile becomes $\mathbf x_1$ and the solution becomes $(a,b)$. Then the agent at $d$ in $\mathbf x_3$ strictly decreases their cost, giving a contradiction. Hence, it must be  $|[0,b]|=|[0,b']|$, i.e., $b=b'$. 

Furthermore, it is easy to see that $a=a'$, as otherwise all of the agents may collude together. It indicates that $f(\mathbf x_3)=f(\mathbf x_1)$. Let us suppose what may happen if $a\ne a'$. If $a'>a$, then agents in $\mathbf{x}_3$ will misreport locations to $\mathbf{x}_1$ so that all agents at $-d$ can reduce their costs and others will not increase costs. If $a'<a$, then agents in $\mathbf{x}_1$ will misreport locations to $\mathbf{x}_3$ so that all agents at $-d$ can reduce their costs and others will not increase costs.

 By a symmetric analysis on profile $\mathbf x_2$ instead of $\mathbf x_3$, we can conclude that  $f(\mathbf x_2)=f(\mathbf x_1)$. For profile $\mathbf{x}_2$, the optimal solution is $(0, d)$, and for profile $\mathbf{x}_3$, the optimal solution is $(-d, 0)$. The optimal social cost are both equal to $\left(\frac{n}{2}-1\right) \cdot\epsilon + d$.
The approximation ratio of $f$ is at least the maximum of the two ratios for profile $\mathbf{x}_3$ and profile $\mathbf{x}_2$, which is 
\begin{align*}
&\frac{\max\left(\text{SC}(\mathbf x_3,f(\mathbf x_3)),\text{SC}(\mathbf x_2,f(\mathbf x_2)) \right)}{\left(\frac{n}{2}-1\right) \cdot\epsilon + d}\nonumber\\
= ~&\frac{\max\left(d-b+\frac{n}{2}\cdot (d+a), d+a+\frac{n}{2}\cdot (d-b)\right)}{\left(\frac{n}{2}-1\right) \cdot\epsilon + d}\\
 = ~&\frac{d+\frac n2d+\max\left(\frac n2a-b,a-\frac n2b\right)}{\left(\frac{n}{2}-1\right) \cdot\epsilon + d}\\
 \ge ~&\frac{d+\frac n2d+(\frac n2+1)(a-b)/2}{\left(\frac{n}{2}-1\right) \cdot\epsilon + d}\\
  \ge ~&\frac{d+\frac n2d-(\frac n2+1)d/2}{\left(\frac{n}{2}-1\right) \cdot\epsilon + d}
  \rightarrow \frac{n+2}{4}.
\end{align*}

\textbf{Case 2}. $a>0$. Obviously, the cost of the agent at $d$ in $\mathbf x_1$ is larger than 0.  Suppose all of the $n$ agents in $\mathbf x_1$ collude to form a location profile where all agents are located at $d$. Since the optimal social cost for this new profile is 0, by the approximation ratio, the solution returned by $f$ must be $(0,d)$. Therefore, this manipulation decreases the cost of the agents at $d$ in $\mathbf x_1$, and the cost of the agents at $-d$ does not change, which contradicts the strong group strategyproofness.

\textbf{Case 3}. $b<0$. This case is symmetric with Case 2.
\end{proof}

\section{The Maximum Cost}\label{sec:max}

In this section, we consider the maximum cost. In Section \ref{sec:optmc}, we characterize the optimal maximum cost and show that the mechanism that returns an optimal solution is not strategyproof. In Section \ref{sec:maxdec}, we present a group strategyproof mechanism that is 2-approximation for the maximum cost. We also show that no deterministic mechanism can beat this ratio by proving a matching lower bound. 

\subsection{The Optimal Maximum Cost}\label{sec:optmc}

We first characterize the optimal maximum cost.
When all agents are located on one side of the facility,  an optimal solution is easy to find by setting one endpoint of the solution on the facility location $0$. When there are agents on both sides, i.e., $x_l<0$ and $x_r>0$, an optimal solution balances the cost of the ``extreme'' agents at $x_l$ and $x_r$ so that they have equal costs. 

\begin{proposition}\label{prop:optmax}
The optimal maximum cost is $\max(0, x_r-d)$ if $x_l\ge 0$, and $\max(0, -x_l-d)$  if $x_r\le 0$. When $x_l < 0 < x_r$, the optimal maximum cost is 
$$
\begin{cases}
   x_r - d, &\text{~if~} x_l + x_r > d;\\
   -x_l - d, &\text{~if~} x_l + x_r < -d;\\
   \max\left(0, \frac{x_r - x_l - d}{2}\right), &\text{~if~} -d \le x_l + x_r \le d.
\end{cases}
$$
\end{proposition}

\begin{proof}
    When $x_l\ge 0$, an optimal solution is $(0,d)$, and the maximum cost is attained by the agent at $x_r$, that is, $MC(\mathbf x,0,d)=\max(0, x_r-d)$.  When $x_r\le 0$, an optimal solution is $(-d,0)$, and the maximum cost is attained by the agent at $x_l$, that is, $MC(\mathbf x,-d,0)=\max(0, -x_l-d)$. 

    Then we consider the case when $x_l<0<x_r$. When $x_l+x_r>d$,  the distance between the agent at $x_r$ and the facility is so large that even under the solution $(0,d)$, the cost of the agent at $x_r$ is still larger than the cost of the agent at $x_l$. Hence, the maximum cost is always achieved by the agent at $x_r$, and the optimal maximum cost is $x_r-d$.  When $x_l+x_r<-d$, by symmetry, the maximum cost is always achieved by the agent at $x_l$, and the optimal maximum cost is $|x_l|-d$. 
    
When $ -d \le x_l + x_r \le d$, if $x_r-x_l \le d$, then $(x_l,x_r)$ is a feasible solution in which all agents have zero cost because everyone as well as the facility lies between the two extreme locations $x_l$ and $x_r$, indicating that the optimal maximum cost is 0.  If $x_r-x_l>d$, to minimize the maximum cost, an optimal solution should guarantee that the agents at $x_l$ and $x_r$ have equal cost, {as otherwise the larger cost of them is the maximum cost and we can move the range toward that agent to decrease it.} This situation is attained by  the solution $(x_l+c_0, x_r-c_0)$ with $c_0=\frac{x_r-x_l-d}{2}$, where both extreme agents have a cost equal to $c_0$.
\end{proof}

 The mechanism that returns an optimal solution is not strategyproof. Consider the location profile where one agent is at $-d$, and the other agent is at $d$. Such a mechanism returns the optimal solution $(-\frac d2,\frac d2)$, and both agents have a cost equal to $\frac d2$. Suppose that the agent at $d$ misreports the location as $1.4d$. Then, the outcome of this mechanism becomes $(-0.3d,0.7d)$, i.e., the unique optimal solution for the new instance. The cost of the agent at $d$ decreases to $0.3d$ after misreporting, giving a contradiction.

\subsection{GSP Mechanisms}\label{sec:maxdec}

We give a simple GSP mechanism that achieves 2-approximation for the maximum cost. 

\begin{mechanism}\label{mec:max1}
    Given location profile $\mathbf x$, return an interval $(\min(x_l,0),\min(x_l,0)+d)$.
\end{mechanism}

\begin{lemma}\label{lemma:mc1-1}
    Mechanism \ref{mec:max1} is group strategyproof.
\end{lemma}

\begin{proof}
    First, note that the agent located at $x_l$ has no incentive to misreport because this agent achieves the best possible cost under the solution
$(\min(x_l, 0), \min(x_l, 0) + d)$. When other agents misreport their locations simultaneously, the solution would be either the same or to the left of the original solution, which cannot benefit any of these agents.
\end{proof}

\begin{theorem}\label{theorem:2}
    Mechanism \ref{mec:max1} is group strategyproof and $2$-approximation for the maximum cost.
\end{theorem}

\begin{proof}
 We only need to prove the approximation ratio.   If $x_l \ge 0$,  the mechanism returns $(0,d)$, and every agent attains their minimum possible cost, implying that the solution is optimal. If  $x_r \le 0$, then  the maximum cost is always attained by the agent at $x_l$, and  the solution $(x_l,x_l+d)$ returned by the mechanism is also optimal. We discuss different cases when $x_l < 0 < x_r$, and the mechanism returns $(x_l,x_l+d)$. Denote by OPT the optimal maximum cost. 

    \textbf{Case 1}. $x_l + x_r > d$. The optimal solution is $(0,d)$, and the optimal maximum cost is 
    $\text{OPT} = \text{MC}(\mathbf x, 0, d) = x_r - d$.
  Under the solution  $(x_l,x_l+d)$ returned by the mechanism,  the maximum cost must be attained by either the agent at $x_l$ or the agent at $x_r$. We have
    \begin{align*}c(x_r, x_l, x_l + d) &= x_r - \max(0, x_l + d) > x_r - d>-x_l\\
        &\ge \max(0, -x_l - d) = c(x_l, x_l, x_l + d).
    \end{align*}
Therefore, the maximum cost is
    \begin{align*}
       \text{MC}(\mathbf x,x_l, x_l + d) &=c(x_r, x_l, x_l + d)= x_r - \max(0, x_l + d)\\ &\le x_r - (x_l+d)
        \le x_r - x_l - d + (x_r - d)\\
        &= 2x_r - x_l - 2d
        \le 2(x_r - d)= 2\cdot \text{OPT}.
    \end{align*}
    It follows the 2-approximation of the mechanism in this case.

    \textbf{Case 2}. $x_l + x_r < -d$. The optimal solution is $(-d,0)$, and the optimal maximum cost is
    $\text{OPT} = \text{MC}(\mathbf x,-d, 0) = -x_l - d$.
  Under the solution  $(x_l,x_l+d)$ returned by the mechanism, the cost of the agent at $x_l$ is no less than the cost of the agent at $x_r$, that is, 
    \begin{align*}
        c(x_l, x_l, x_l + d) &= -x_l - d \ge \max(0, x_r - d)\\
        &= c(x_r, x_l, x_l + d).
    \end{align*}
Therefore, the maximum cost is
    \begin{align*}
        \text{MC}(\mathbf x,x_l, x_l + d) &=c(x_l, x_l, x_l + d)\\
        &= -x_l - d = \text{OPT}.
    \end{align*}
  Thus, the mechanism is optimal for the maximum cost in this case.

    \textbf{Case 3}. $-d \le x_l + x_r \le d$. By Proposition \ref{prop:optmax}, the  optimal maximum cost is $ \max\left(0, \frac{x_r - x_l - d}{2}\right)$. Precisely,  if $x_r-x_l \le d$, then every agent has zero cost under the solution $(x_l,x_l+d)$ returned by the mechanism. The maximum cost induced by the mechanism is 0, and thus it is optimal in this case. 
    
    If $x_r-x_l>d$, an optimal solution that minimizes the maximum cost guarantees that the agents at $x_l$ and $x_r$ have equal cost, which is equal to $\frac{x_r-x_l-d}{2}$.
  Under the solution  $(x_l,x_l+d)$ returned by the mechanism, the maximum cost is achieved by either the agent at $x_r$ or the agent at $x_l$. We note that
    \begin{align*}
        ~& c(x_r, x_l, x_l + d) = x_r - \max(0, x_l + d) \\
        \ge ~& d - x_l - \max(0, x_l + d) \ge \max(0, -x_l - d),
    \end{align*}
where  the last inequality comes from the discussion over two subcases: if $x_l + d\ge 0$, then
    \begin{align*}
        d-x_l-\max(0, x_l + d) &= -2x_l > 0= \max(0, -x_l - d),
    \end{align*}
    and if $x_l + d < 0$, then
    \begin{align*}
        d-x_l-\max(0, x_l + d) &= d - x_l> -d - x_l \\
        &= \max(0, -x_l - d).
    \end{align*}
Combining with  $c(x_l,x_l,x_l+d)=\max(0, -x_l - d)$, we have
    \begin{align*}
       \text{MC}(\mathbf x,x_l, x_l + d) &=c(x_r, x_l, x_l + d) = x_r - \max(0, x_l + d) \\
        &\le x_r - x_l - d = 2\cdot \frac{x_r - x_l - d}{2}= 2\cdot \text{OPT}.
    \end{align*}
    Therefore the mechanism is 2-approximation in this case.
\end{proof}

The analysis on the approximation ratio of Mechanism  \ref{mec:max1}  is tight. Consider the location profile $\mathbf x=(-d, d)$. It is clear that the optimal solution is $(\frac d2,\frac d2)$, and the optimal maximum cost is $\frac d2$, attained by both extreme agents. The solution returned by the mechanism is $(-d, 0)$, inducing a maximum cost equal to $c(x_r,-d,0)=c(d,-d,0)=d$, which is 2 times of the optimal maximum cost.

Next, we complete the results by proving a tight lower bound for all (group) strategyproof mechanisms. 

\begin{theorem}\label{theorem:3}
    For the maximum cost, no deterministic strategyproof mechanism has an approximation ratio of less than 2.
\end{theorem}

\begin{proof}
    Let $f$ be any deterministic strategyproof mechanism.  Consider a location profile $\mathbf x=(-\frac d2-\epsilon,\frac d2+\epsilon)$, where $\epsilon>0$ is a sufficiently small value. Under any solution, there is at least one agent whose cost is at least $\epsilon$. Assume w.l.o.g. that under the solution $f(\mathbf x)$, the cost of agent 2 is at least $\epsilon$.

    Next, suppose that agent 2 misreports the location as $x_2'=d$, and the location profile becomes $\mathbf x'=(-\frac d2-\epsilon,d)$. Let $f(\mathbf x')=(a,b)$ be the outcome of the mechanism. By the strategyproofness of mechanism $f$, the cost of agent 2 cannot decrease, and we have $c(\frac d2+\epsilon,a,b)\ge \epsilon$.
        It immediately implies that 
    $$\left|\left[a,b\right]\cap [0,\frac d2+\epsilon]\right|= \frac d2+\epsilon - c\left(\frac d2+\epsilon,a,b\right) \le \frac d2.$$
    Then, since $b-a\le d$, we have either $b\le\frac d2$, or $a\ge \epsilon$. 
    
    For the location profile $\mathbf x'=(-\frac d2-\epsilon,d)$, if $b\le\frac d2$, the maximum cost of $\mathbf{x}'$ induced by the mechanism is at least $c(d,a,b)\ge \frac d2$. 
    If $a\ge \epsilon$, the maximum cost induced by the mechanism is at least $c(-\frac d2-\epsilon,a,b)\ge \frac d2+2\epsilon$. Hence, in both cases, the induced maximum cost induced by mechanism $f$ on location profile $\mathbf{x}'$  is at least $\frac d2$. On the the hand,  the optimal solution for $\mathbf x'$ is $(-\frac d4-\frac{\epsilon}{2},\frac{3d}{4}-\frac{\epsilon}{2})$, and the optimal maximum cost is $\frac d4+\frac{\epsilon}{2}$. Therefore, the approximation ratio approaches 2 when $\epsilon\rightarrow 0$. 
\end{proof}

Combining Theorem \ref{theorem:2} and Theorem \ref{theorem:3}, we can conclude that Mechanism  \ref{mec:max1}  is the best possible (group) strategyproof mechanism we could expect for the maximum cost.

\subsection{SGSP Mechanisms}

We notice that Mechanism \ref{mec:max1} is not strong group strategyproof. Consider the location profile $\mathbf x=(-2d,\ldots,-2d,\frac{-d}{2},\ldots,\frac{-d}{2})$ where $\frac n2$ agents are located at $-2d$, and the other $\frac n2$ agents are located at $\frac{-d}{2}$. Since $x_l=-2d$, the mechanism returns a range $(-2d,-d)$. The cost of agent at $-2d$ is $d$, and the cost of the agent at $\frac{-d}{2}$ is $\frac{d}{2}$. Suppose that all of the $n$ agents misreport their location as $-d$ simultaneously so that the reported location profile becomes $\mathbf x=(-d,-d,\ldots,-d)$. Then $x_l'=-d$, and the mechanism returns the solution $(-d,0)$. The cost of the agent at $-2d$ remains to be $d$, but the cost of the agent at $\frac{-d}{2}$ decreases to 0, giving a contradiction to the strong group strategyproofness.

Let us recall  Mechanism \ref{mec:soc-strong}, which is proven to be SGSP in Lemma \ref{lemma:strong1}. Note that the property of strategyproofness does not rely on the objective but only on the agent cost functions. Thus, Mechanism \ref{mec:soc-strong} is SGSP here.  In the following, we show that the two mechanisms have indeed the same approximation ratio for the maximum cost. 

\begin{theorem}\label{thm:strong3}
    Mechanism \ref{mec:soc-strong} is $2$-approximation for the maximum cost.
\end{theorem}
\begin{proof}
Since Mechanism \ref{mec:max1}  is 2-approximation and {the only difference between Mechanism \ref{mec:max1}  and  Mechanism  \ref{mec:soc-strong} lies in the case when $x_l\le -d$,  } we only need to discuss this case. 
When $x_l<-d$, Mechanism \ref{mec:max1} returns $(x_l,x_l+d)$, and  Mechanism  \ref{mec:soc-strong} returns $(-d,0)$. By Lemma \ref{lemma:basic1}, the cost of any agent under solution $(-d,0)$ is at most that under $(x_l,x_l+d)$. 
 Therefore,  Mechanism \ref{mec:soc-strong} is also $2$-approximation for the maximum cost.
\end{proof}

\section{Inapproximability Results for Randomized Mechanisms}\label{sec:ext}

In this section, we consider the extent to which randomization can be used to design strategyproof mechanisms with better approximation ratios. To this end, we consider randomized strategyproof mechanisms. A randomized mechanism $f: \mathbb{R}^n \rightarrow \Delta \left(\mathbb{R}^2 \right)$ maps a location profile $\mathbf{x}$ to a probability distribution, where $\Delta \left(\mathbb{R}^2 \right)$ is the set of all possible probability distributions of outcomes in $\mathbb{R}^2$.
Given a distribution $f(\mathbf x)\in\Delta(\mathbb R^2)$, the cost of each agent $i$ is the expectation $\mathbb E_{(a,b)\sim f(\mathbf x)}[c(x_i,a,b)]$. 

Below, we show that any randomized strategyproof mechanisms cannot be much better than our deterministic mechanisms in terms of approximation ratios.

\begin{theorem}\label{thm:999}
    No randomized SGSP mechanism has an approximation ratio less than $\frac{n+2}{4}$ for the social cost. 
\end{theorem}

\begin{proof}
Let $f$ be a randomized SGSP mechanism. We consider three location profiles $\mathbf{x}_1, \mathbf{x}_2, \mathbf{x}_3$ with 
even number $n$ of agents. In $\mathbf{x}_1$, $\frac{n}{2}$ agents are located at $-d$ and the other $\frac{n}{2}$ are located at $d$. In $\mathbf{x}_2$, one agent is located at $-d$, $\frac{n}{2}-1$ agents are located at 0, and other $\frac{n}{2}$ are located at $d$. In $\mathbf{x}_3$, one agent is located at $d$, $\frac{n}{2}-1$ agents are located at 0 and other $\frac{n}{2}$ are located at $-d$. 

Let the distributions returned by mechanism $f$ be $f(\mathbf{x}_1)=(a,b), f(\mathbf{x}_2) = (a',b'), f(\mathbf{x}_3) = (a'',b'')$. Suppose $P_1 = \textbf{Pr}(a>0), P_2 = \textbf{Pr}(b<0)$, and $P_0=1-P_1-P_2 = \textbf{Pr}(a\le 0\wedge b\ge 0)$. The corresponding events are denoted by $A_1, A_2, A_0$.

Due to SGSP property, these agents cannot misreport from $\mathbf{x}_1$ to $\mathbf{x}_2$ to decrease the cost of agent at $-d$, that is, we have $\mathbb{E}[c(-d,a,b)]\le \mathbb{E}[c(-d,a',b')]$. This immediately implies that $\mathbb{E}\left[|[-d, 0] \cap [a, b]|\right] \ge \mathbb{E}\left[|[-d, 0] \cap [a', b']|\right]$. 

Also, these agents  also cannot misreport from $\mathbf{x}_2$ to $\mathbf{x}_1$. Then $\mathbb{E}\left[|[-d, 0] \cap [a, b]|\right] \le \mathbb{E}\left[|[-d, 0] \cap [a', b']|\right]$ holds by contradiction proof. Suppose $\mathbb{E}\left[|[-d, 0] \cap [a, b]|\right] > \mathbb{E}\left[|[-d, 0] \cap [a', b']|\right]$, then according to SGSP, we must also have $\mathbb{E}\left[|[0, 0] \cap [a', b']|\right] > \mathbb{E}\left[|[0, 0] \cap [a, b]|\right]= 0$. This leads to a contradiction.

Therefore we know $\mathbb{E}\left[|[-d, 0] \cap [a, b]|\right] = \mathbb{E}\left[|[-d, 0] \cap [a', b']|\right]$ according to previous inequalities. It is also easy to know that $\mathbb{E}\left[|[0, d] \cap [a, b]|\right] = \mathbb{E}\left[|[0, d] \cap [a', b']|\right]$, as otherwise all of the agents may collude together.

For $\mathbf{x}_1$ and $\mathbf{x}_3$, $\mathbb{E}\left[|[0, d] \cap [a, b]|\right] = \mathbb{E}\left[|[0, d] \cap [a'', b'']|\right]$ and $\mathbb{E}\left[|[-d, 0] \cap [a, b]|\right] = \mathbb{E}\left[|[-d, 0] \cap [a'', b'']|\right]$ hold by symmetry. In this way, we know the following two equations:
\begin{align*}
     \mathbb{E}[|[-d, 0]\cap [a, b]|] &= \mathbb{E}[|[-d, 0]\cap [a', b']|] \\
     &= \mathbb{E}[|[-d, 0]\cap [a'', b'']|]; \\
     \mathbb{E}[|[0, d]\cap [a, b]|] &= \mathbb{E}[|[0, d]\cap [a', b']|] \\
     &= \mathbb{E}[|[0, d]\cap [a'', b'']|]. \\
\end{align*}

We then calculate the social cost of profile $\mathbf{x}_3$ and profile $\mathbf{x}_2$. We have
\begin{align*}
    & \text{SC}(\mathbf{x}_3, f(\mathbf{x}_3)) \\
    = ~&\frac{n+2}{2}\cdot d - \frac{n}{2}\cdot \mathbb{E}[|[a'', b'']\cap [-d, 0]|] - \mathbb{E}[|[a'', b'']\cap [0, d]|] \\
    = ~&\frac{n+2}{2}\cdot d - \frac{n}{2}\cdot \mathbb{E}[|[a, b]\cap [-d, 0]|] - \mathbb{E}[|[a, b]\cap [0, d]|] \\
    \ge ~& P_0 \left(\frac{n}{2} \mathbb{E}[d - [-d, 0]\cap [a, 0] | A_0] + \mathbb{E}[d - [0, d]\cap [0, b] | A_0]\right) \\
    ~&  + P_1\cdot \frac{n}{2}\cdot \mathbb{E}[d - 0 | A_1] + P_2\cdot \mathbb{E}[d - 0 | A_2] \\
    = ~&P_0\cdot\left(\frac{n}{2}\cdot \left(d + \mathbb{E}[a | A_0]\right) + d - \mathbb{E}[b | A_0] \right) + \left(P_1\cdot \frac{n}{2} + P_2\right)d
    \end{align*}
and 
 \begin{align*}
 & \text{SC}(\mathbf{x}_2, f(\mathbf{x}_2))\\
    = ~&\frac{n+2}{2}\cdot d - \frac{n}{2}\cdot \mathbb{E}[|[a', b']\cap [0, d]|] - \mathbb{E}[|[a', b']\cap [-d, 0]|] \\
    = ~&\frac{n+2}{2}\cdot d - \frac{n}{2}\cdot \mathbb{E}[|[a, b]\cap [0, d]|] - \mathbb{E}[|[a, b]\cap [-d, 0]|] \\
    \ge ~& P_0 \left(\frac{n}{2} \mathbb{E}[d - [0, d]\cap [0, b] | A_0] + \mathbb{E}[d - [-d, 0]\cap [a, 0] | A_0]\right) \\
     ~& + P_2\cdot \frac{n}{2}\cdot \mathbb{E}[d - 0 | A_2] + P_1\cdot \mathbb{E}[d - 0 | A_1] \\
    = ~& P_0\cdot\left(\frac{n}{2}\cdot \left(d - \mathbb{E}[b | A_0]\right) + d + \mathbb{E}[a | A_0] \right) + \left(P_2\cdot \frac{n}{2} + P_1\right)d. \\
\end{align*}

The approximation ratio of $f$ is at least the maximum of the two ratios for profile $\mathbf{x}_3$ and profile $\mathbf{x}_2$. The optimal social cost for both instances is $d$. Then we have
\begin{align*}
    & \max\left(\text{SC}(\mathbf{x}_3, f(\mathbf{x}_3)), \text{SC}(\mathbf{x}_2, f(\mathbf{x}_2))\right) \\
    \ge ~& \frac{1}{2}\cdot \left(\text{SC}(\mathbf{x}_3 , f(\mathbf{x}_3)) + \text{SC}(\mathbf{x}_2 , f(\mathbf{x}_2))\right) \\
    \ge ~& P_0\cdot \left(\frac{n+2}{2}\cdot d - \frac{n+2}{4}\cdot (\mathbb{E}[b|A_0] - \mathbb{E}[a|A_0]) \right) +\frac{n+2}{4}\cdot (P_1 + P_2) \cdot d\\
    \ge ~& \frac{n+2}{4}\cdot (P_0 + P_1 + P_2)\cdot d\\
    = ~& \frac{n+2}{4} \cdot \text{OPT},
\end{align*}
where the last inequality comes from the fact that $\mathbb{E}[b|A_0] - \mathbb{E}[a|A_0]\le d$ since in every realization of the distribution we have $b-a\le d$. Therefore, the approximation ratio of $f$ is at least $\frac{n+2}{4}$. 

\end{proof}

\begin{theorem}\label{theorem:4}
    For the maximum cost, no randomized strategyproof mechanism has an approximation ratio of less than 1.5.
\end{theorem}

\begin{proof}
    Let $f$ be any randomized strategyproof mechanism with an approximation ratio of $1.5-\delta$ for some $\delta>0$.  Consider a location profile $\mathbf x=(-\frac d2-\epsilon,\frac d2+\epsilon)$, where $\epsilon>0$ is a sufficiently small value. Under any solution, there is at least one agent whose expectation of cost is at least $\epsilon$. Assume w.l.o.g. that under the solution $f(\mathbf x)$, the expectation of cost of agent 2 is at least $\epsilon$.

  Next, suppose that agent 2 misreports the location as $x_2'=\frac{d}{2}+3\epsilon$, and the location profile becomes $\mathbf x'=(-\frac d2-\epsilon,\frac{d}{2}+3\epsilon)$. Let $f(\mathbf x), f(\mathbf x')$ be the outcome distributions of the mechanism, and we use $(a, b) \sim f(\mathbf x), (a', b') \sim f(\mathbf x')$. 

  By the strategyproofness of mechanism $f$, the cost of agent 2 cannot decrease, and we have $\mathbb E\left[c(\frac d2+\epsilon, (a',b'))\right]\ge \epsilon$. It immediately implies that 
  \begin{align*}
      \mathbb E\left[\left|[a',b']\cap \left[0,\frac d2+\epsilon\right]\right|\right] &= \mathbb{E}\left[\frac d2+\epsilon - c\left(\frac d2+\epsilon,(a',b')\right)\right] \le \frac d2
  \end{align*}

  Let $P_0=\textbf{Pr}(a'>0)$,  $P_3=\textbf{Pr}(b'<0)$ and $P_4=\textbf{Pr}(a'\le 0,b'\ge 0)$. It is clear that $P_0+P_3+P_4=1$. Define $A$ to be the event $a'\le 0\wedge b'\ge 0$. Then we have 
 \begin{align*}
     P_4 \cdot \mathbb E\left[\left.\min\left(b', \frac{d}{2}+\epsilon\right)\right|A\right]&=P_4\cdot \mathbb E\left[\left.\left|[0, b']\cap \left[0, \frac{d}{2}+\epsilon\right]\right|\right|A\right]\\
     &\le \mathbb E\left[\left|[a',b']\cap \left[0,\frac d2+\epsilon\right]\right|\right]\le \frac{d}{2}
 \end{align*}

  Then we calculate the expectation of the maximum cost with $x_l=-\frac{d}{2}-\epsilon, x_r=\frac{d}{2}+3\epsilon$. We use $ \text{MC}$ instead of $ \text{MC}(\mathbf x',f(\mathbf x'))$ for simplicity.  Since the optimal maximum cost is $2\epsilon$, we have

  \begin{align}
      \text{MC} &=P_0\cdot \mathbb E[\text{MC}|a'>0]+P_3\cdot \mathbb E[\text{MC}|a'\le 0,b'<0]\nonumber+P_4\cdot \mathbb E[\text{MC}|A]\nonumber\\
      &\ge P_0\left(\frac d2+\epsilon\right)+P_3\left(\frac d2+3\epsilon\right)\nonumber+P_4\cdot\mathbb E\left[\left.\max\left(a'+\frac{d}{2}+\epsilon, \frac{d}{2}+3\epsilon-b'\right)\right|A\right]\nonumber \\
      &\ge P_0\left(\frac d2+\epsilon\right)+P_3\left(\frac d2+3\epsilon\right)\nonumber+P_4\cdot\mathbb E\left[\left.\max\left(b'-\frac{d}{2}+\epsilon, \frac{d}{2}+3\epsilon-b'\right)\right|A\right]\nonumber \\
      &\ge (1-P_4)\left(\frac d2+\epsilon\right)+P_4\cdot\mathbb E\left[\left.\max\left(b'-\frac{d}{2}+\epsilon, \frac{d}{2}+3\epsilon-b'\right)\right|A\right]\label{eq:555}
  \end{align}

  Since  $\epsilon$ is a sufficiently small number and $\frac d2+\epsilon \gg 1.5\text{OPT}=3\epsilon$,  it must be $P_4\rightarrow 1$. Precisely, it follows from $(1-P_4)\left(\frac d2+\epsilon\right)<1.5OPT$ that $P_4>1-\frac{3\epsilon}{d/2+\epsilon}$. 

  Define $P_1=\textbf{Pr}(b' \le \frac{d}{2}+\epsilon|A)$, and $P_2 = 1 - P_1=\textbf{Pr}(b' > \frac{d}{2} + \epsilon|A)$. Define $\mathbf{b_{1}'}=\mathbb E\left[\min\left(b', \frac{d}{2}+\epsilon\right)|b' \le \frac{d}{2}+\epsilon,A\right]$ to denote the weighted average  $b'$ with $b'\le \frac{d}{2}+\epsilon$. Therefore we have
  \begin{align}
        P_4  P_1\cdot \mathbf{b_{1}'} + P_4P_2\cdot \left(\frac{d}{2}+\epsilon\right) &= P_4\cdot\mathbb E\left[\left.\min\left(b', \frac{d}{2}+\epsilon\right)\right|A\right]\le \frac{d}{2}\label{eq:666}
    \end{align}

  By \eqref{eq:555},  we have
\begin{align*}
    &\frac{\text{MC}}{P_4}-\left(\frac{1}{P_4}-1\right)\left(\frac d2+\epsilon\right) \\
    &\ge \mathbb E\left[\left.\max\left(b'-\frac{d}{2}+\epsilon, \frac{d}{2}+3\epsilon-b'\right)\right|A\right]\\
    &= P_1\cdot\mathbb E\left[\left.\frac{d}{2}+3\epsilon-b'\right|b'\le \frac{d}{2}+\epsilon,A\right] + P_2\cdot\mathbb E\left[\left.b'-\frac{d}{2}+\epsilon\right|b'> \frac{d}{2}+\epsilon,A\right]\\
      &= P_1\cdot\left(\frac{d}{2}+3\epsilon-\mathbb E\left[b'|b'\le \frac{d}{2}+\epsilon,A\right]\right) + P_2\cdot \left(\mathbb E\left[b'|b'> \frac{d}{2}+\epsilon,A\right]-\frac{d}{2}+\epsilon\right)\\
      &\ge P_1\cdot\left(\frac{d}{2}+3\epsilon-\mathbf{b_{1}'}\right) + P_2\cdot 2\epsilon\\
      &= P_1\cdot\left(\frac{d}{2}+3\epsilon\right) + P_2\cdot 2\epsilon -P_1\cdot \mathbf{b_{1}'}\\
      &\ge P_1\cdot\left(\frac{d}{2}+3\epsilon\right) + P_2\cdot 2\epsilon +P_2\cdot \left(\frac{d}{2}+\epsilon\right) -\frac{d}{2P_4} \text{\quad (By (\ref{eq:666}))}\\
      &= P_1\cdot\left(\frac{d}{2}+3\epsilon\right) +P_2\cdot \left(\frac{d}{2}+3\epsilon\right) -\frac{d}{2P_4}\\
      &=3\epsilon + \frac d2\left(1-\frac{1}{P_4}\right)
\end{align*}

Since $P_4>1-\frac{3\epsilon}{d/2+\epsilon}$ and $\epsilon$ is sufficiently small compared with $d$ and $\delta$, it follows that 
\begin{align*}
    \text{MC}&\ge P_4\cdot 3\epsilon+(1-P_4)\epsilon >\left(3-\frac{6\epsilon}{d/2+\epsilon}\right)\epsilon >(1.5-\delta)\cdot 2\epsilon =(1.5-\delta)\cdot \text{OPT},
\end{align*}
 which is a contradiction to the approximation ratio $1.5-\delta$ of mechanism $f$.
\end{proof}

\section{Conclusion}

We investigate a variation of facility location problems (FLPs) on real lines to improve the accessibility of agents to a prelocated facility by constructing accessibility ranges (or intervals)
to extend the accessibility of the facility within
the context of mechanism design without money.  
We proposed several (asymptotically) tight group strategyproof and strong group strategyproof deterministic mechanisms that minimize the social cost and the maximum cost of the agents. 
We also provide randomized lower bounds of any randomized strategyproof mechanisms, showing that our deterministic GSP/SGSP mechanisms are reasonable even when considering randomization.

\paragraph{Acknowledgements.}
Hau Chan is supported by the National Institute of General Medical Sciences of the National Institutes of Health [P20GM130461], the Rural Drug Addiction Research Center at the University of Nebraska-Lincoln, and the National Science Foundation under grants IIS:RI \#2302999 and IIS:RI \#2414554. 
Chenhao Wang is supported in part by the Guangdong Provincial Key Laboratory of Interdisciplinary Research and Application for Data Science, UIC, project code 2022B1212010006, by UIC research grant R0400001-22, and by  Artificial Intelligence and Data Science Research Hub, UIC, No. 2020KSYS007 and No. UICR0400025-21. 
 Chenhao Wang is also supported by NSFC under grant 12201049, and by UIC under grants  UICR0400014-22, UICR0200008-23 and UICR0700036-22. The content is solely the responsibility of the authors and does not necessarily represent the official views of the funding agencies.

\bibliographystyle{plain}
\bibliography{mybibfile}

\end{document}